\documentclass[11pt, twoside, a4paper]{article}

\usepackage{amsmath}
\usepackage{amsthm}  
\usepackage[paper=a4paper,dvips,top=3cm,left=3cm,right=3cm,foot=3.5cm,bottom=4.5cm]{geometry}
\usepackage{amsfonts}
\usepackage{amssymb}
\usepackage{bbm}

\newtheorem{theorem}{Theorem}
\newtheorem{lemma}{Lemma}

\newcommand {\ind} {\mathbbm{1}}

\newcommand {\reals} {{\rm I\!R}}

\newcommand {\bs} {\mbox{\boldmath $s$}}

\newcommand {\bu} {\mbox{\boldmath $u$}}

\newcommand {\bx} {\mbox{\boldmath $x$}}
\newcommand {\by} {\mbox{\boldmath $y$}}

\newcommand {\bU} {\mbox{\boldmath $U$}}

\newcommand {\bX} {\mbox{\boldmath $X$}}

\newcommand{\calA}{{\cal A}}

\newcommand{\calC}{{\cal C}}
\newcommand{\calD}{{\cal D}}
\newcommand{\calE}{{\cal E}}

\newcommand{\calR}{{\cal R}}
\newcommand{\calS}{{\cal S}}
\newcommand{\calT}{{\cal T}}
\newcommand{\calU}{{\cal U}}

\newcommand{\calX}{{\cal X}}
\newcommand{\calY}{{\cal Y}}

\newcommand {\hP} {\hat{P}}  
\newcommand {\tP} {\tilde{P}}  %

\newcommand {\tby} {\tilde{\boldsymbol{y}}}
\newcommand{\eqde}{\triangleq}

\def\be{\begin{eqnarray}}
\def\ee{\end{eqnarray}}
\def\ben{\begin{eqnarray*}}
\def\een{\end{eqnarray*}}

\begin{document}

\thispagestyle{empty}
\title{Achievable Error Exponents for Channel with Side Information -- Erasure and List Decoding}
\author{Erez Sabbag and Neri Merhav}

\maketitle

\begin{center}
Department of Electrical Engineering \\
Technion - Israel Institute of Technology \\
Technion City, Haifa 32000, Israel\\
{\tt \{erezs@tx, merhav@ee\}.technion.ac.il}
\end{center}
\vspace{1.5\baselineskip}
\setlength{\baselineskip}{1.5\baselineskip}

\begin{abstract}
We consider a decoder with an erasure option and a variable size list decoder for channels with non-casual side information at the transmitter. First, universally achievable error exponents are offered for decoding with an erasure option using a parameterized decoder in the spirit of Csisz\'{a}r and K\"{o}rner's decoder. Then, the proposed decoding rule is generalized by extending the range of its parameters to allow variable size list decoding. This extension gives a unified treatment for erasure/list decoding. Exponential bounds on the probability of  list error and the average number of incorrect messages on the list are given. Relations to Forney's and Csisz\'{a}r and K\"{o}rner's decoders for discrete memoryless channel are discussed. These results are obtained by exploring a random binning code with conditionally constant composition codewords proposed by Moulin and Wang, but with a different decoding rule.

\end{abstract}
\section{Introduction}
\label{sec.intro}

A decoder with an erasure option is a decoder which has the option of not deciding, i.e., to declare an ``erasure''. On the other hand, a variable size list decoder is a decoder which produces a list of estimates for the correct message rather than a single estimate, where a list error occurs when the correct message is not on the list. In \cite{Forney68}, Forney explored the random coding error exponents of erasure/list decoding for discrete memoryless channels (DMC's). These bounds were obtained by analyzing the optimal decoding rule \cite[eq.~(11)]{Forney68}
\be
\label{Forney_dec}
\by \in \calR_m \; \textrm{iff} \; \Pr(\by,\bx_m) \ge e^{NT} \sum_{m' \neq m} \Pr(\by,\bx_{m'})
\ee
where $\Pr(\by,\bx_m)$ is the joint probability of the channel output $\by$ and the codeword $\bx_m$, and $T$ is an arbitrary parameter.
The bounds were obtained using Gallager's bounding techniques. Forney showed that the list option and the erasure option are ``two sides of the same coin'', namely, by changing the value of $T$ one can switch from list decoding ($T$ is negative) to decoding with an erasure option ($T$ is positive).

In \cite[Th.~5.11]{CK81}, Csisz\'{a}r and K\"{o}rner derived universally achievable error exponents for a decoder with an erasure option for DMC's. These error exponents were obtained by analyzing the following universal decoding rule \cite[p.~176]{CK81} for constant composition (CC) codes:
\be
\label{csiszar_dec}
\varphi(\by) = \left\{\begin{array}{lll} m &,&  I(\bx_m;\by) > \tilde{R} + \lambda |I(\bx_{m'};\by)- R|^+ \quad \forall m' \neq m \\
0 &,& \textrm{else} \end{array} \right.
\ee
where $R$ is the code rate (i.e., $m \in \{1,\dots,2^{NR} \}$), $\bx_m$ is a codeword taken from a given type class $T_{\bx}$, $\by$ is the channel output, $I(\bx;\by)$ is the empirical mutual information, and
$\tilde{R} \ge R$ and $\lambda>0$ are arbitrary parameters.
This decoding rule generalizes the maximum mutual information (MMI) decoder \cite[p.~164]{CK81} to include an erasure option.
The bounds were obtained using a fixed composition coding and by applying the packing lemma derived in \cite[Lemma~5.1]{CK81}. However, these bounds were not extended to variable size list decoding. We note that the decoding rule \eqref{csiszar_dec} depends on the coding rate $R$, which might limit its generality.
Moreover, it was stated that \eqref{csiszar_dec} is an unambiguous decoding rule for $\lambda>0$, a fact that was used to derive the error exponents. It turns out that this decoding rule is unambiguous only when $\lambda \ge 1$. Unlike Forney's decoder \eqref{Forney_dec}, no optimality claims were made for this decoder but, in \cite[Sec.~4.4.3]{TelatarPhD92} Teletar stated that these bounds are ``essentially the same as those in \cite{Forney68}''.

Recently, Moulin \cite{Moulin08} generalized Csisz\'{a}r's decoder using a weighting function:
\ben
\textstyle
\varphi(\by) = \left\{\begin{array}{l} m \:,\:  I(\bx_m;\by) > R + \max_{m' \neq m} F\big(I(\bx_{m'};\by)- R\big)  \\
\textstyle
0 \:,\: \textrm{else} \end{array} \; , \right.
\een
where $F(\cdot)$ is a continuous, non-decreasing function. The corresponding error exponents were analyzed and it it was shown that for some rates and channels these error exponents coincide with Forney's error exponents. Note that Moulin's proposed decoder is a function of the code rate $R$ similarly to Csisz\'{a}r's decoder.

In \cite{TelatarGallager94},\cite{Telatar98} Teletar and Gallager proposed tighter exponential bounds on decoding with an erasure option and list decoding for DMC using the method of types. These bounds are not universal in general since the decoding metric depends on the channel statistics. However, it is claimed that under certain conditions these bounds are tighter than Forney's bounds. See \cite[Sec.III]{Telatar98}.

As far as we know, no similar bounds were ever offered for discrete memoryless channels with random states, which are observed by the encoder but not by the decoder \cite{GelPin80}. For ordinary decoding (without erasure/list option), Moulin and Wang \cite{MoulinWang07} recently derived an achievable error exponent for channels with state information present non-causally at the transmitter. These results were obtained by analyzing the error probability of a stacked binning scheme and a maximum penalized mutual information (MPMI) decoder.

In this work, we use the random code construction proposed by Moulin and Wang \cite{MoulinWang07} to derive achievable error exponents for decoding with an erasure option and variable size list decoding.  In Section~\ref{sec.notation}, we propose a parameterized decoding rule with an erasure option in the spirit of \eqref{csiszar_dec}.
In Section~\ref{sec.erasure}, we derive universally achievable error exponents by analyzing the proposed decoding rule. In Section~\ref{sec.list}, achievable error exponents are offered to decoder with a list option. These exponents are obtained by extending the range of the proposed decoder's parameters to allow decoding with a list option. The generalized decoding rule enables a unified treatment for erasure/list decoding similar to Forney's decoder \eqref{Forney_dec}. In Section~\ref{sec.discussion}, relations to Forney's and Csisz\'{a}r and K\"{o}rner's decoders for DMC are discussed. Moreover, it is shown that the obtained error exponents generalize some known results.

\section{Notation and Preliminaries}
\label{sec.notation}

We begin with some notations and definitions. Throughout this work, capital
letters represent scalar random variables (RVs\label{p.RV}), and specific
realizations of them are denoted by the corresponding lowercase letters. Random
vectors of dimension $N$ will be denoted by bold-face letters. The notation
$\mathbbm{1}\{A\}$, where $A$ is an event, will designate the indicator
function of $A$ (i.e.,$\mathbbm{1}\{A\}=1$ if $A$ occurs and
$\mathbbm{1}\{A\}=0$ otherwise).
The notion $a_n \doteq b_n$, for two positive sequences $\{a_n\}_{n \ge 1}$ and $\{b_n\}_{n \ge 1}$, expresses asymptotic equality in the logarithmic scale, i.e.,
$$
\lim_{n \to \infty} \frac{1}{n} \ln \left(\frac{a_n}{b_n}\right)=0.
$$
Let the vector $\hat{P}_{\bx}=\big\{\hat{P}_{\bx}(a), \; a \in \calX \big\}$ denote the empirical distribution induced by a vector $\bx \in \calX^n$, where $\hP_{\bx}(a) = \frac{1}{n}{\sum_{i=1}^n \mathbbm{1}\{x_i=a\}}$. The type class $T_{\bx}$ is the set of vectors $\tilde{\bx} \in \calX^n$ such that $\hat{P}_{\tilde{\bx}}=\hat{P}_{\bx}$. A type class induced by the empirical distribution $\hP_{\bx}$ will be denoted by $T(\hP_{\bx})$.
Similarly, the joint empirical distribution induced by $(\bx,\by)\in \calX^n \times \calY^n$ is the vector
$\hat{P}_{\bx\by}=\left\{\hat{P}_{\bx\by}(a,b),~a\in\calX,~b\in\calY \right\}$ where
\ben
\hat{P}_{\bx\by}(a,b)=\frac{1}{n}\sum_{i=1}^n\mathbbm{1}\big\{x_i=a,y_i=b\big\},~~~x\in\calX,~y\in\calY \;,
\een
i.e., $\hat{P}_{\bx\by}(a,b)$ is the relative frequency of the pair $(a,b)$
along the pair sequence $(\bx,\by)$. Likewise, the type class $T_{\bx\by}$ is the set of all pairs $(\tilde{\bx},\tby) \in \calX^n \times \calY^n$ such that $\hat{P}_{\tilde{\bx}\tilde{\by}}=\hat{P}_{\bx\by}$. The conditional type class $T_{\by|\bx}$, for  given vectors $\bx \in \calX^n$, and $\by \in \calY^n$ is the set of all vectors $\tby \in \calY^n$ such that $T_{\bx\tby}=T_{\bx\by}$. The Kullback-Leibler divergence between two distributions $P$ and $Q$ on $\calA$, where $|\calA| < \infty$  is defined as
$$
\calD(P \| Q) = \sum_{a \in \calA} P(a) \ln \frac{P(a)}{Q(a)} \;,
$$
with the conventions that $0 \ln 0 =0$, and $p \ln \frac{p}{0}=\infty$ if $p >0$. We denote the empirical entropy of a vector $\bx \in \calX^n$ by $\hat{H}(\bx)$, where $\hat{H}(\bx)= - \sum_{a \in \calX} \hat{P}_{\bx}(a) \ln \hat{P}_{\bx}(a)$. Other information theoretic quantities governed by empirical distributions (e.g., conditional empirical entropy, empirical mutual information) will be denoted similarly.
Finally, we define $|t|^+ \eqde \max\{0,t\}$ and $\exp_2(t) \eqde 2^t$.

Consider a discrete memoryless state dependent channel with a finite input alphabet $\calX$, a finite state alphabet $\calS$, a finite output alphabet $\calY$, and a probability transition distribution $W(y|x,s)$. Given an input sequence $\bx$ and a state sequence $\bs$ emitted from a discrete memoryless source $P_S(\bs)=\prod_{i=1}^N P_S(s_i)$, the channel output sequence $\by$ is generated according to the conditional distribution
$W(\by|\bx,\bs) = \prod_{i=1}^N W(y_i|x_i,s_i)$.
A message $m \in \{1,\ldots,M\}$ is to be transmitted to the receiver. We assume that the state sequence $\bs$ is available at the transmitter non-causally, but not at the receiver. We also assume that all messages are a-priori equiprobable.
Given $\bs$ and $m$, the transmitter produces a sequence $\bx = f_N(\bs,m)$ which is used to convey message $m$ to the decoder.

\subsection{Codebook construction \cite{MoulinWang07}}
\label{subsec.code}
In \cite[p.~1337]{MoulinWang07}, Moulin and Wang used in their derivation a binning code with conditionally constant composition (CCC) codewords. This code will be used in our proofs. For the sake of completeness, we briefly describe the code construction and the encoding process. The decoding part will be described in detail later. The code construction requires the use of an auxiliary random variable $U \in \calU$ which takes on values in a finite set of size $|\calX||\calS|+1$. See \cite[Sec.~III.E]{MoulinWang07} for more information.

For a given empirical conditional distribution $\hP^*_{\bx\bu|\bs}$, a sub-code $\calC(\hP_{\bs})$ is constructed for each state sequence type class $T_{\bs}=T(\hP_{\bs})$. 
Given a state type class $T(\hP_{\bs})$, compute the marginal distribution
$$
\hP^*_{\bu}(u) = \sum_{x} \sum_{s} \hP^*_{\bx\bu|\bs}(x,u|s) \hP_{\bs}(s)
$$
where $\hP_{\bs}$ is the empirical distribution induced $T_{\bs}$. Note that $\hP^*_{\bu}(u)$ is a function of $\hP_{\bs}$ and it might be different for a different state type class. Draw $2^{N (R+\rho(\hP_{\bs}))}$ random vectors independently from the type class $T^*_U(\hP_{\bs})$ induced by $\hP^*_{\bu}$, according to uniform distribution, where $\rho(\cdot)$ is a general bin-depth function. Arrange the vectors in an array with $M=2^{NR}$ columns and $2^{N \rho(\hP_{\bs})}$ rows. The code $\calC$ is the union of all sub-codes, i.e., $\calC = \bigcup_{\hP_{\bs}} \calC(\hP_{\bs})$. Note that the number of these sub-codes is polynomial in $N$. In this work, we choose $\rho(\hP_{\bs})=I^*_{US}(\hP_{\bs})+\epsilon$ , where
$$
I^*_{US}(\hP_{\bs})= \sum_{u,s} \hP_{\bs}(s)P^*_{\bu|\bs}(u|s) \log \frac{\hP_{\bs}(s)P^*_{\bu|\bs}(u|s)}{\hP^*_{\bu}(u)} \;,
$$
i.e., $I^*_{US}(\hP_{\bs})$ is the mutual information $I(U;S)$ induced by $\hP_{\bs}(S)\cdot \hP^*_{\bu|\bs}(U|S)$, and $\epsilon$ is an arbitrarily small positive constant.
This choice ensures that the probability of encoding error vanishes at a double-exponentially rate \cite[p.~1338]{MoulinWang07}.

The encoding of message $m$  given a state sequence $\bs$ is done in two steps: (i) ~Find an index $l$ such that $\bu_{l,m} \in \calC(\hP_{\bs})$ is a member of the conditional type class  $T^*_{\bu_{l,m}|\bs} = \{\bu' : \hP_{\bu'\bs} = \hP^*_{\bu_{l,m}|\bs}\hP_{\bs} \}$. If more than one such $l$ exists, pick one at random under the uniform distribution. If no such $l$ can be found, pick $\bu$ at random from $T^*_{\bu|\bs}$ under the uniform distribution. (ii) ~Draw $\bX$ uniformly from $T^*_{\bx|\bu_{l,m} \bs}$, induced by $\hP^*_{\bx\bu|\bs}$ and $(\bu_{l,m},\bs)$. For notational simplicity, we use the shorthand $\lambda$ to denote the type of state sequences $\hP_{\bs}$, and $\bu_{\lambda,l,m}$ to denote $\bu_{l,m} \in \calC(\lambda)$.

In \cite{MoulinWang07}, a maximum penalized mutual information (MPMI) decoder was used to decode the above code. A MPMI decoder seeks a vector $\bu \in \calC$ that maximizes the penalized empirical mutual information criterion $\max_{\hP_{\bs}} \max_{\bu \in \calC(\hP_{\bs})} \big[I(\bu;\by) - \psi(\hP_{\bs}) \big]$, where $\psi(\cdot)$ is a general penalty function. It was shown that the optimal choice of these functions is $\rho(\hP_{\bs})=\psi(\hP_{\bs})=I^*_{US}(\hP_{\bs})+\epsilon$ where $\epsilon$ is an arbitrarily small positive constant.
In this work, we assume that  $\psi(\hP_{\bs})=I^*_{US}(\hP_{\bs})=I^*_{\hP_{\bs}P^*_{\bu|\bs}}(U;S)$ for reasons that will be given later.
To allow decoding with an erasure/list option, we propose to modify the MPMI decoding rule in the spirit of \eqref{csiszar_dec}. We choose $\rho(\hP_{\bs})=\psi(\hP_{\bs})=I^*_{US}(\hP_{\bs})+\epsilon$ for reasons that will be discussed in Section~\ref{sec.discussion}.

\subsection{The proposed decoding rule}
For a given code $\calC$ constructed as described in Subsection~\ref{subsec.code}, we propose to use the following decoder $\varphi: \calY^N  \to \{0,1,\ldots,M\}$ with an erasure option: Declare $m$ if
\be
\label{dec_rule1}
I(\bu_{\lambda,l,m};\by)-I^*_{US}(\lambda) > T +  \alpha\big|I(\bu_{\lambda',l',m'};\by)-I^*_{US}(\lambda')\big|^+ \quad \forall m' \neq m, \lambda', l',
\ee
otherwise, declare $0$ (i.e., ``erasure''), where $\alpha \ge 1$ and $T \ge 0$ are arbitrary parameters.

Our first step is to show that this decoder is unambiguous, i.e., at most one message index taken from $\{1,\ldots,M\}$ fulfills \eqref{dec_rule1}. This property is essential to allow decoding with an erasure option. This property is stated in the following Lemma:

\begin{lemma}
\label{lemma_ambiguity}
For $\alpha \ge 1$ and $T \ge 0$, the proposed decoding rule \eqref{dec_rule1} is unambiguous.
\end{lemma}
The proof of Lemma~\ref{lemma_ambiguity} is deferred to the Appendix.
Using a similar proof of Lemma~\ref{lemma_ambiguity}, it can be shown that Csisz\'{a}r and K\"{o}rner's decoder \eqref{csiszar_dec} might be ambiguous if $0 < \lambda < 1$, contrary to the statement made in \cite[Th.~5.11]{CK81}.

\section{Erasure Option}
\label{sec.erasure}

Given a code $\calC$, a decoder with an erasure option is a partition of $\calY^N$ into $(M+1)$ regions $\calR_0,\calR_1,\ldots, \calR_M$.
The decoder decides in favor of message $m$ if $\by \in \calR_m$, $m=1,\ldots,M$, or it declares ``erasure'' if $\by \in \calR_0$. Following Forney \cite{Forney68}, let us define two error events. The event $\calE_1$ is the event in which $\by$ does not fall in the decision region of the correct message. The event $\calE_2$ is the event of \emph{undetected error}, namely, the event in which $\by$ falls in $\calR_{m'}$, $m' \neq 0$, while $m$ was transmitted. The probabilities of these error events are given by
\be
\Pr \{ \calE_1 \} & = & \frac{1}{M} \sum_{m=1}^M \sum_{\by \in \calR_m^c} P(\by|\bx_m) \\
\Pr \{ \calE_2 \} & = & \frac{1}{M} \sum_{m=1}^M \sum_{\by \in \calR_m} \sum_{m' \neq m} P(\by|\bx_{m'})
\ee
where $P(\by | \bx_m) = \sum_{\bs \in \calS^N} P_S(\bs) W\big(\by|\bx_m(\bs),\bs \big)$.

Let $J(P_S P_{UX|S} P_{Y|XS}) \eqde I(U;Y) - I(U;S)$ where $P_S$ , $P_{UX|S}$, and $P_{Y|XS}$ are three (conditional) probability distributions of the quadruplet RVs $(U,S,X,Y)$.
The following theorem presents exponential bounds on $\Pr\{\calE_1\}$ and $\Pr\{\calE_2\}$ for decoding with erasure option:
\begin{theorem}
\label{Thorem_erasure}
For every $\alpha \ge 1$ and $T \ge 0$ there exists a $N-$length block code of rate $R$ such that the following error exponents can be achieved simultaneously
\be
\Pr\{\calE_1\} &\le& \exp_2 \left\{ -N E_1(R,W,T,\alpha) \right\}  \\
\Pr\{\calE_2\} &\le& \exp_2 \left\{ -N E_2(R,W,T,\alpha) \right\}
\ee
where
\begin{multline}
\label{E_1_erasure}
E_1(R,W,T,\alpha) = \min_{ \tP_{S}} \max_{\tP_{UX|S}} \min
\Big\{
\min_{\tP_{Y|XS} : J(\tP_S \tP_{UX|S} \tP_{Y|XS}) \le T} \calD(\tP_{S}\tP_{UX|S}\tP_{Y|XS} \| P_S \tP_{UX|S} W) , \\
\min_{\tP_{Y|XS}} \Big[  \calD(\tP_{S}\tP_{UX|S}\tP_{Y|XS} \| P_S \tP_{UX|S} W)
+ \Big|  \frac{1}{\alpha} \big(J(\tP_S \tP_{UX|S} \tP_{Y|XS}) - T\big)  - R   \Big|^+ \Big]  \Big\}
\end{multline}
and
\begin{multline}
\label{E_2_erasure}
E_2(R,W,T,\alpha) =
 \min_{\tP_{S}} \max_{\tP_{UX|S}} \min_{\tP_{Y|XS}} \Big\{ \\
\calD\big(\tP_{S}\tP_{UX|S}\tP_{Y|XS} \| P_{S}\tP_{UX|S}W\big)
 + \Big| T + \alpha |J(\tP_S \tP_{UX|S} \tP_{Y|XS})|^+ - R   \Big|^+  \Big\}\;.
\end{multline}
\end{theorem}

\begin{proof}
We analyze $\Pr\{\calE_1\}$ and $\Pr\{\calE_2\}$ using the proposed decoder \eqref{dec_rule1}. The proof is similar in some parts to the derivation done in the proof of Theorem~3.2 in \cite{MoulinWang07}, but it is given in full for the sake of completeness.

Fix a probability distribution $\hP_{\bu\bx|\bs}$ and construct a code $\calC$ as described in Subsection~\ref{subsec.code}. %
An encoding error occurs when the first encoding step fails. Namely, given $m$ and $\bs$ there is no index $l$ such that $\bu_{\lambda,l,m} \in T^*_{\bu|\bs}$. Since $\bu$ is drawn randomly according to uniform distribution from $T^*_U(\hP_{\bs})$  it follows that
\be
\Pr \left\{\calE_c(m) | \bs \right\}  = \left[1  - \Pr\left\{ \bU \in T^*_{\bu|\bs} | \bU \in T^*_U(\hP_{\bs}) \right\} \right]^{2^{N \rho(\hP_{\bs})}}
\ee
where $\calE_c(m)$ denotes encoding error when message index $m$ is encoded, and
\be
\Pr\left\{ \bU \in  T^*_{\bu|\bs} | \bU \in T^*_U(\hP_{\bs}) \right \} & = & \frac{|T^*_{\bu|\bs}|}{T^*_U(\hP_{\bs})} \nonumber \\
& \doteq&  2^{- N I^*_{US} (\hP_{\bs}) } \;.
\ee
Since $\rho(\hP_{\bs})$ was chosen to be greater than $I^*_{US}(\hP_{\bs})$ by $\epsilon$ we get that probability of encoding error of message $m$ given $\bs$ is upper bounded by
\be
\label{encode_err}
\Pr \left\{\calE_c(m) | T(\hP_{\bs}) \right\} \le  \exp \big\{-2^{N\epsilon} \big\} \;,
\ee
namely, the probability of encoding error decays in a double-exponential rate. See step~1 in \cite[p.~1338]{MoulinWang07} for more details.

The undetected error probability can be expressed as follows
\be
\label{p_e2a}
\Pr\{\calE_2\} &=& \frac{1}{M} \sum_{m=1}^M \Pr\{\calE_2| m \textrm{ is to be sent}\}\nonumber\\
&=& \Pr \Big\{ \calE_2(1) \big| \: m=1  \Big\} \nonumber\\
&=&  \sum_{T_{\bu\bs\bx\by}} P\big(T_{\bu\bs\bx\by} |\: m=1\big) \Pr \Big\{\calE_2(1) \big| T_{\bu\bs\bx\by}, m=1  \Big\} \nonumber\\
&\le&  \sum_{T_{\bu\bs\bx\by}} P\big(T_{\bu\bs\bx\by}  | \: m=1 \big) \Bigg[\Pr\Big\{ \calE_c(1) \big| T_{\bu\bs\bx\by}, m=1  \Big\}  \nonumber\\
&& \quad + \Pr\Big\{ \calE_2(1) \big| T_{\bu\bs\bx\by}, m=1 , {\calE_c(1)}^c  \Big\} \Bigg] \; ,
\ee
where $\calE_2(1)$ is the event of undetected error given that $m=1$ was sent, $T_{\bu\bs\bx\by}$ is the joint typical class of the quadruplet $(\bu,\bs,\bx,\by)$.
Since all messages are drawn according the same probability distributions, the probability of a type class $T_{\bu\bs\bx\by}$ is independent of the message index $1 \le m \le M$.
Therefore,
\be
\label{P_Tusxy}
P(T_{\bu\bs\bx\by}| \: m =1 ) & = & \exp_2 P(T_{\bu\bs\bx\by} )  \nonumber \\
& \doteq & \exp_2 \{ -N \calD(\hP_{\bs} \hP_{\bu\bx|\bs} \hP_{\by|\bx\bs} \| P_S \hP_{\bu\bx|\bs} W) \} \;,
\ee
as was shown in \cite[eq.(5.12)]{MoulinWang07}. An undetected error \emph{can} occur only if  there is a $\bu_{\lambda',l',m'} \in \calC$ such that
\be
\label{e2_era}
I(\bu_{\lambda',l',m'};\by)-I^*_{US}(\lambda') > T + \alpha|I(\bu;\by)-I^*_{US}(\hP_{\bs}) |^+  \;,
\ee
conditioned on $\bu,\bs,\by$ and $T_{\bu\bs\bx\by}$.

Following \cite[eq.~(5.13)]{MoulinWang07}, the undetected error probability is upper bounded by
\be
\Pr\Big\{ \calE_2(1) \big| T_{\bu\bs\bx\by}, m=1 , {\calE_c(1)}^c  \Big\} =
  1- \prod_{\hP_{\bs'}} \Big[1- P_{e_2} (\bu, \by, \hP_{\bs'}, T_{\bu\bs\bx\by})\Big]^{2^{NI^*_{US}(\hP_{\bs'})}(2^{NR}-1)} \;,
\ee
where $P_{e_2} (\bu, \by, \hP_{\bs'}, T_{\bu\bs\bx\by})$ is the probability that for some $l'$ and $m' \neq 1$,  $\bu_{l',m'} \in \calC(\hP_{\bs'})$ fulfills \eqref{e2_era} conditioned on $\bu,\by$ and $T_{\bu\bs\bx\by}$. $P_{e_2} (\bu, \by, \hP_{\bs'}, T_{\bu\bs\bx\by})$ can be expressed as follows
\ben
P_{e_2} (\bu, \by, \hP_{\bs'}, T_{\bu\bs\bx\by}) = \sum_{\bu' \in \calU_{e_2}(\bu, \by,\hP_{\bs'},\hP_{\bu\bs\bx\by})} P(\bu'|\hP_{\bs'}) \nonumber \\
= \sum_{\bu' \in \calU_{e_2}(\bu, \by,\hP_{\bs'},\hP_{\bu\bs\bx\by})}  |T^*_U(\hP_{\bs'})|^{-1}
\een
where
\be
\calU_{e_2}(\bu, \by,\hP_{\bs'},\hP_{\bu\bs\bx\by}) =
 \Big\{ \bu' \in T_U^*(\hP_{\bs'})\: : \: I(\bu';\by) - I^*_{US}(\hP_{\bs'}) > T  + \alpha \big|I(\bu;\by) - I^*_{US}(\hP_{\bs}) \big|^+ \Big\} \;. \nonumber \\
\ee
Clearly, $\calU_{e_2}(\bu, \by,\hP_{\bs'},\hP_{\bu\bs\bx\by})$ is contained in the following set of conditional types
\be
\calT_{e_2}(\bu, \by,\hP_{\bs'},\hP_{\bu\bs\bx\by}) & = &
\Big\{ T_{\bu'|\by}  \; : T_{\bu'}=T^*_U(\hP_{\bs'}), \nonumber \\
&&  \qquad I(\bu';\by) - I^*_{US}(\hP_{\bs'}) > T + \alpha \big|I(\bu;\by) - I^*_{US}(\hP_{\bs})\big|^+ \Big\} \nonumber \\
 &\subseteq&  \Big\{ T_{\bu'|\by}  \; :  I(\bu';\by) - I^*_{US}(\hP_{\bs'}) > T + \alpha \big| I(\bu;\by) - I^*_{US}(\hP_{\bs}) \big|^+ \Big\} \nonumber \;.
\ee
Using similar steps as in \cite[eq.~(5.14)-(5.17)]{MoulinWang07}, we get
\be
\label{pe2_bound}
&& P_{e_2} (\bu, \by, \hP_{\bs'}, T_{\bu\bs\bx\by}) \nonumber \\
&& \le \sum_{T_{\bu'|\by} \subseteq \calT_{e_2}(\bu, \by,\hP_{\bs'},\hP_{\bu\bs\bx\by})} \frac{|T_{\bu'|\by}|}{|T_{\bu'}|} \nonumber \\
&& \doteq \sum_{T_{\bu'|\by} \subseteq \calT_{e_2}(\bu, \by,\hP_{\bs'},\hP_{\bu\bs\bx\by})} 2^{-N I(\bu';\by)} \nonumber \\
&& \le  \exp_2 \Big\{-N \Big[T + \alpha |I(\bu;\by) - I^*_{US}(\hP_{\bs})|^+ +I^*_{US}(\hP_{\bs'}) \Big] \Big\} \;.
\ee
Applying the following bound  \cite[eq.~(5.18)]{MoulinWang07}, which can be regarded as a generalized union bound,
\be
\label{unionBound}
1 - \prod_i(1-\alpha_i)^{t_i} \le \ \min \left\{ 1 , \sum_i \alpha_i t_i \right\}, \qquad 0 \le \alpha_i \le 1 , t_i \ge 1
\ee
on Eq.\eqref{pe2_bound} we get
\be
\label{p_e2b}
 \Pr \left\{ \calE_2(1) | T_{\bu\bs\bx\by} , {\calE_c(1)}^c \right\}
&=&  1- \prod_{\hP_{\bs'}} \Big[1- P_{e_2} (\bu, \by, \hP_{\bs'}, T_{\bu\bs\bx\by})\Big]^{2^{NI^*_{US}(\hP_{\bs'})}(2^{NR}-1)}  \nonumber \\
& \le &  \min \left\{1, \sum_{\hP_{\bs'}} P_{e_2}(\bu,\by,T_{\bu\bs\bx\by}) 2^{N I^*_{US}(\hP_{\bs'})} (2^{NR}-1) \right\} \nonumber \\
&\le& \exp_2 \Big\{ -N \Big| T + \alpha |I(\bu;\by) - I^*_{US}(\hP_{\bs})|^+  - R   \Big|^+  \Big\} \:.
\ee
Combining \eqref{p_e2a}, \eqref{P_Tusxy}, \eqref{p_e2b} and optimizing over $\hP_{\bu\bx|\bs}$ and $\hP_{\bs}$ we get that
\be
\Pr\{\calE_2\} \le \exp_2 \left\{ -N E_2(R,W,T,\alpha) \right\} \;,
\ee
where $\calE_2(R,W,T,\alpha)$ is given in \eqref{E_2_erasure}.

Similarly to derivation of $E_2(R,W,T,\alpha)$, we can upper bound the probability of not making the right decision, denoted by $\Pr\{\calE_1\}$. This error event occurs when the received $\by$ does not belong to the decision region corresponding to the transmitted message $m$. Therefore, an error occurs when
\be
I(\bu;\by)-I^*_{US}(\hP_{\bs}) \le T+\alpha|I(\bu_{\lambda',l',m'};\by)-I^*_{US}(\lambda')|^+
\ee
for some $\lambda',l'$ and $m' \neq m$, conditioned on $\bu,\bs,\by$ and $T_{\bu\bs\bx\by}$.  This happens if and only if
\be
\label{1st_cond}
I(\bu;\by)-I^*_{US}(\lambda) \le T \;,
\ee
or
\be
\label{e1_cond2}
T < I(\bu;\by)-I^*_{US}(\lambda) \le T+\alpha|I(\bu_{\lambda',l',m'},\by)-I^*_{US}(\lambda')|^+.
\ee
Following \eqref{e1_cond2}, $I(\bu;\by)-I^*_{US}(\lambda) - T$ is strictly positive since  $T \ge 0$ and $I(\bu;\by)-I^*_{US}(\lambda)$ is strictly greater than $T$. Moreover, \eqref{e1_cond2} implies that $I(\bu_{\lambda',l',m'},\by)-I^*_{US}(\lambda')$ must be positive too since
$$
T < T+\alpha|I(\bu_{\lambda',l',m'},\by)-I^*_{US}(\lambda')|^+  \quad \Rightarrow \quad
\alpha|I(\bu_{\lambda',l',m'},\by)-I^*_{US}(\lambda')|^+ > 0 , \;\; a \ge 1
$$
which means that the clipping function $| \cdot |^+$ was not active, namely, $I(\bu_{\lambda',l',m'},\by)-I^*_{US}(\lambda') > 0$. Therefore \eqref{e1_cond2} implies that
\be
\label{e1_era}
\frac{1}{\alpha} \Big[I(\bu;\by) - I^*_{US}(\hP_{\bs}) -T \Big] \le I(\bu_{\lambda',l',m'};\by) - I^*_{US}(\lambda') \;.
\ee
Hence, the event of not making the right decision is a union of two disjoint events \eqref{1st_cond} and \eqref{e1_era}. Therefore,
\be
\label{err_e1}
\Pr\{\calE_1\} &=& \frac{1}{M} \sum_{m=1}^M \Pr \Big\{ \calE_1 \big |  m \textrm{ is to be sent}  \Big\}  \nonumber\\
& = & \Pr \Big\{ \calE_1(1) \big |  m =1 \Big\} \nonumber\\
& = &  \sum_{T_{\bu\bs\bx\by}} P\big(T_{\bu\bs\bx\by}  \big) \Pr \Big\{\calE_1(1) \big| T_{\bu\bs\bx\by}, m=1  \Big\} \nonumber\\
&\le&   \sum_{T_{\bu\bs\bx\by}} P\big(T_{\bu\bs\bx\by}  \big) \Big[\Pr\Big\{ \calE_c(1) \big| T_{\bu\bs\bx\by}, m=1  \Big\}
 + \ind \big\{I(\bu;\by)-I^*_{US}(\hP_{\bs}) \le T \big\}   \nonumber\\
&& \qquad + \Pr\Big\{ \calA(1) \big| T_{\bu\bs\bx\by}, m=1 , {\calE_c(1)}^c  \Big\} \Big]  \;,
\ee
where $\calE_1(1)$ is the event of making the wrong decision given that $m=1$,
and $\calA(1)$ is the event in which $T < I(\bu;\by) - I^*_{US}(\hP_{\bs}) \le T + \alpha[I(\bu_{\lambda',l',m'};\by) - I^*_{US}(\lambda')]$ for some $\lambda',l'$ and $m' \neq 1$. The last sum can be rewritten as follows
\be
\label{err_e1_2}
\sum_{T_{\bu\bs\bx\by}} P\big(T_{\bu\bs\bx\by}  | \: m=1 \big) \Big[\Pr\Big\{ \calE_c(1) \big| T_{\bu\bs\bx\by}, m=1  \Big\}
 + \Pr\Big\{ \calA(1) \big| T_{\bu\bs\bx\by}, m=1 , {\calE_c(1)}^c  \Big\} \Big] \nonumber\\
 \qquad + \sum_{T_{\bu\bs\bx\by}} P\big(T_{\bu\bs\bx\by}  | \: m=1 \big) \ind \big\{I(\bu;\by)-I^*_{US}(\hP_{\bs}) \le T \big\} \;. \nonumber \\
\ee
The second summand of \eqref{err_e1_2} can easily be estimated using the method of types and by applying \eqref{P_Tusxy}
\be
\label{pe1_firstSum}
\sum_{T_{\bu\bs\bx\by}} && P\big(T_{\bu\bs\bx\by} | \: m =1 \big) \ind \big\{I(\bu;\by)-I^*_{US}(\hP_{\bs}) \le T \big\} = \nonumber \\
& = & \sum_{T_{\bu\bs\bx\by}: I(\bu;\by)-I^*_{US}(\hP_{\bs}) \le T } P\big(T_{\bu\bs\bx\by} \big) \nonumber \\
&\doteq&  \exp_2  \left\{-N \min_{\hP_{\bu\bs\bx\by}: I(\bu;\by)-I^*_{US}(\hP_{\bs}) \le T }
\calD(\hP_{\bs}\hP_{\bu\bx|\bs}\hP_{\by|\bx\bs} \| P_S \hP_{\bu\bx|\bs}W) \right\}. \;
\ee
As for the first summand of \eqref{err_e1}, it can be upper bounded similarly to the undetected error probability in the following way
\be
\Pr\Big\{ \calA(1) \big| T_{\bu\bs\bx\by}, m=1 , {\calE_c(1)}^c  \Big\}
= 1- \prod_{\hP_{\bs'}} \Big[1- P_{e_1} (\bu, \by, \hP_{\bs'}, T_{\bu\bs\bx\by})\Big]^{2^{NI^*_{US}(\hP_{\bs'})}(2^{NR}-1)}
\ee
where $P_{e_1} (\bu, \by, \hP_{\bs'}, T_{\bu\bs\bx\by})$ is the probability that $\bu_{l',m'} \in \calC(\hP_{\bs'})$ fulfills \eqref{e1_era} conditioned on $\bu,\by$ and $T_{\bu\bs\bx\by}$ for some $l'$ and $m'\neq 1$. Therefore,
\be
P_{e_1} (\bu, \by, \hP_{\bs'}, T_{\bu\bs\bx\by}) = \sum_{\bu' \in \calU_{e_1}(\bu, \by,\hP_{\bs'},\hP_{\bu\bs\bx\by})} P(\bu'|\hP_{\bs'}) \nonumber \\
= \sum_{\bu' \in \calU_{e_1}(\bu, \by,\hP_{\bs'},\hP_{\bu\bs\bx\by})}  |T^*_U(\hP_{\bs'})|^{-1}
\ee
and
\begin{multline}
\calU_{e_1}(\bu, \by,\hP_{\bs'},\hP_{\bu\bs\bx\by}) = \\ \Big\{ \bu' \in T_U^*(\hP_{\bs'})\; :
I(\bu';\by) - I^*_{US}(\hP_{\bs'})
>\frac{1}{\alpha} [I(\bu;\by) - I^*_{US}(\hP_{\bs}) - T] \Big\} \;.
\end{multline}

Again, $\calU_{e_1}(\bu, \by,\hP_{\bs'},\hP_{\bu\bs\bx\by})$ is contained in the following set of conditional types
\be
\calT_{e_1}(\bu, \by,\hP_{\bs'},\hP_{\bu\bs\bx\by}) & = &
\Big\{ T_{\bu'|\by}  \; : T_{\bu'}=T^*_U(\hP_{\bs'}), \nonumber \\
&&  \qquad I(\bu';\by) - I^*_{US}(\hP_{\bs'}) > \frac{1}{\alpha} \big[I(\bu;\by) - I^*_{US}(\hP_{\bs}) -T\big] \Big\} \nonumber \\
 &\subseteq&  \Big\{ T_{\bu'|\by}  \; :  I(\bu';\by) - I^*_{US}(\hP_{\bs'}) > \frac{1}{\alpha} \big[ I(\bu;\by) - I^*_{US}(\hP_{\bs}) - T \big] \Big\} \nonumber \;.
\ee
Using similar steps as in the first part of the proof, we get that
\be
\label{pe1_bound}
&& P_{e_1} (\bu, \by, \hP_{\bs'}, T_{\bu\bs\bx\by}) \nonumber \\
&& \le \sum_{T_{\bu'|\by} \subseteq \calT_{e_1}(\bu, \by,\hP_{\bs'},\hP_{\bu\bs\bx\by})} \frac{|T_{\bu'|\by}|}{|T_{\bu'}|} \nonumber \\
&& \doteq \sum_{T_{\bu'|\by} \subseteq \calT_{e_1}(\bu, \by,\hP_{\bs'},\hP_{\bu\bs\bx\by})} 2^{-N I(\bu';\by)} \nonumber \\
&& \le  \exp_2 \Big\{-N \Big[ \frac{1}{\alpha} \big( I(\bu;\by) - I^*_{US}(\hP_{\bs})-T\big) +I^*_{US}(\hP_{\bs'}\Big] \Big\} \;.
\ee
Applying the union bound \eqref{unionBound}, we get
\be
\label{p_e1b}
 \Pr \left\{ \calE_1(1) | T_{\bu\bs\bx\by} , {\calE_c(1)}^c \right\}
&=&  1- \prod_{\hP_{\bs'}} \Big[1- P_{e_1} (\bu, \by, \hP_{\bs'}, T_{\bu\bs\bx\by})\Big]^{2^{NI^*_{US}(\hP_{\bs'})}(2^{NR}-1)}  \nonumber \\
&\le& \exp_2 \Big\{ -N \Big| \frac{1}{\alpha} \big[I(\bu;\by) - I^*_{US}(\hP_{\bs})-T\big] - R \Big|^+  \Big\} \:.
\ee
Combining \eqref{err_e1}, \eqref{pe1_firstSum}, \eqref{P_Tusxy},\eqref{p_e1b}, and optimizing over $\hP_{\bu\bx|\bs}$ and $\hP_{\bs}$ we get that
\be
\Pr\{\calE_1\} \le \exp_2 \left\{ -N E_1(R,W,T,\alpha) \right\} \;,
\ee
where $E_1(R,W,T,\alpha)$ is given in \eqref{E_1_erasure}.

\end{proof}

\section{List Decoding}
\label{sec.list}

A decoder with a variable size list produces a list of candidate estimates for the correct message. Let $\calE_1$ denote the list error event, namely, the event in which $\by$ does not fall in the decision region corresponding to the correct message. As stated by Forney \cite[p.~206]{Forney68}, the event corresponding to $\calE_2$ under decoding with an erasure option, is the average number of incorrect messages on the list, denoted by $\bar{N}_I$, where
\be
\bar{N}_I = \sum_{m'} \Pr\{m' \; \textrm{is on the list and incorrect} \} \;.
\ee

In the following Theorem, exponential bounds are offered for the probability of $\calE_1$ and on the average number of incorrect messages on the list. These bounds are obtained by generalizing the proposed decoding rule \eqref{dec_rule1} to the variable list size case by extending the range of its parameters $\alpha$ and $T$.
\begin{theorem}
\label{Thorem_list}
For every $\alpha \in (0,1)$ and $T \in \reals$ there exists a $N-$length block code of rate $R$ such that the following error exponents can be achieved simultaneously
\be
\Pr\{\calE_1\} &\le& \exp_2 \left\{ -N E_1(R,W,T,\alpha) \right\}  \\
\bar{N}_I &\le& \exp_2 \left\{ -N E_2(R,W,T,\alpha) \right\} \;,
\ee
where
\begin{multline}
\label{E_1_list}
E_1(R,W,T,\alpha) = \min_{\tP_{S}} \max_{\tP_{UX|S}} \min
\Bigg\{
  \min_{\tP_{Y|XS} : J(\tP_S \tP_{UX|S} \tP_{Y|XS}) \le T} \calD(\tP_{S}\tP_{UX|S}\hP_{Y|XS} \| P_S \tP_{UX|S} W) , \\
\min_{\tP_{Y|XS}} \Big(  \calD(\tP_{Y|XS} \| W | \tP_{S}\tP_{UX|S})
+ \Big|  \frac{1}{\alpha} \big[J(\tP_S \tP_{UX|S} \tP_{Y|XS}) - T\big]  - R   \Big|^+  \Big)  \Bigg\}  \;,
\end{multline}
and
\begin{multline}
\label{E2list}
E_2(R,W,T,\alpha) =\min_{\tP_{S}} \max_{\tP_{UX|S}} \min_{\tP_{Y|XS}} \Bigg\{
\calD\big(\tP_{S}\tP_{UX|S}\tP_{Y|XS} \| P_{S}\tP_{UX|S}W\big)  \\
+ \big| T + \alpha |J(\tP_S \tP_{UX|S} \tP_{Y|XS})|^+ \big|^+ - R \Bigg\} \;.
\end{multline}
\end{theorem}

\begin{proof}
To allow list option we take $\alpha \in (0,1)$ and $T \in \reals$. Therefore, the following decoding rule will be used:
add $m$ to the list if
$$
I(\bu_{\lambda,l,m};\by)-I^*_{US}(\lambda) > T+ \alpha|I(\bu_{\lambda',l',m'};\by)-I^*_{US}(\lambda')|^+
$$
for all  $\lambda', l',m' \neq m$. An empty list is regarded as  ``erasure''.

Fix a probability distribution $\hP_{\bu\bx|\bs}$ and construct a code $\calC$ as described in Subsection~\ref{subsec.code}. The encoding error is described in \eqref{encode_err}.
Let us start with the probability of list error.
A list error occurs when
$$
I(\bu;\by)-I^*_{US}(\hP_{\bs}) \le T + \alpha\big|I(\bu_{\lambda',l',m'};\by)-I^*_{US}(\lambda')\big|^+
$$
for some $l',\lambda'$ and $m' \neq m$, conditioned on $\bu, \bs, \by$ and $T_{\bu\bs\bx\by}$. This happens if
$$
I(\bu;\by)-I^*_{US}(\lambda) \le T
$$
or
$$
T < I(\bu;\by)-I^*_{US}(\lambda) \le T  +\alpha[I(\bu_{\lambda',l',m'},\by)-I^*_{US}(\lambda')] \;,
$$
similarly to the second part of the proof of Theorem~\ref{Thorem_erasure}.
From this point on we follow the derivation of $\Pr\{\calE_1\}$ in the proof of Theorem~\ref{Thorem_erasure}, and obtain the desired exponent \eqref{E_1_list}.

Our next step is to upper bound the average number of incorrect words on the list $\bar{N}_I$.
Following \cite[eq.~(12)-(13)]{Forney68}, the average number of incorrect codewords is
\be
\label{N_bar}
\bar{N}_I & = & \frac{1}{M} \sum_{m=1}^M \sum_{m' \neq m} \Pr \Big\{ m' \: \textrm{is on list} \big|\: m \: \textrm{was sent}  \Big\} \nonumber \\
&  = & \sum_{m' > 1} \Pr \Big\{ m' \: \textrm{is on list} \big|\: m=1  \Big\} \nonumber \\
& = & (M-1) \sum_{T_{\bu\bs\bx\by}} P(T_{\bu\bs\bx\by} | \: m=1) \Pr \Big\{ m' \neq m \: \textrm{is on list} \big|\: m=1 , T_{\bu\bs\bx\by}  \Big\}  \nonumber \\
\ee
where the second equality is because the messages are equiprobable. The probability that $m' > 1$ is on the list given that $m=1$ was sent can be bounded as follows
\be
\Pr \Big\{ m' \: \textrm{is on list} \big|\: m=1  \Big\} \le \Pr\Big\{ \calE_c(1) \big| T_{\bu\bs\bx\by}, m=1  \Big\} \nonumber \\
 \qquad \quad  + \Pr\Big\{ m' \neq m \: \textrm{is on list} \big| T_{\bu\bs\bx\by}, m=1 , {\calE_c(1)}^c  \Big\}
\ee
where $\calE_c(1)$ is the encoding error event of message $m=1$. Applying \eqref{P_Tusxy} we get
\be
\bar{N}_I& \le & M  \sum_{T_{\bu\bs\bx\by}} P(T_{\bu\bs\bx\by} )
\Big[\Pr\Big\{ \calE_c(1) \big| T_{\bu\bs\bx\by}, m=1  \Big\} \nonumber \\
&& \qquad \quad  + \Pr\Big\{ m' \neq m \: \textrm{is on list} \big| T_{\bu\bs\bx\by}, m=1 , {\calE_c(1)}^c  \Big\} \Big] \;.
\ee

The probability that $m' \neq 1$ is on the decoding list given that $m=1$ was sent successfully and given $T_{\bu\bs\bx\by}$ is upper bounded by
\begin{multline}
\label{list1}
\Pr\Big\{ m' \neq m \: \textrm{ is on list} \big| T_{\bu\bs\bx\by}, m=1 , {\calE_c(1)}^c  \Big\} = \\
1- \prod_{\hP_{\bs'}} \Big[1- P_{e_2} (\bu, \by, \hP_{\bs'}, T_{\bu\bs\bx\by})\Big]^{2^{NI^*_{US}(\hP_{\bs'})}},
\end{multline}
where $P_{e_2} (\bu, \by, \hP_{\bs'}, T_{\bu\bs\bx\by})$ is the probability that there exist $\bu_{l',m'} \in \calC(\hP_{\bs'})$ which obeys the decoding rule conditioned on $\bu,\by$ and $T_{\bu\bs\bx\by}$ for some $l'$.
Namely, there exist a codeword $\bu_{l',m'}$ which ``beats'' all other codewords from a different column.
This probability is upper bounded by the event in which $\bu_{l',m'} \in \calC(\hP_{\bs'})$ ``beats'' \emph{only} the correct codeword, i.e.,
\ben
I(\bu_{\lambda',l',m'};\by)-I^*_{US}(\hP_{\bs'}) > T + \alpha\big|I(\bu;\by)-I^*_{US}(\lambda)\big|^+ \:,
\een
which in turn upper bounds \eqref{list1}. Therefore,

From this point on we follow the derivation of $\Pr\{\calE_2\}$ in the proof of Theorem~1. Note that $M$ multiplies $\Pr\Big\{ m' \neq m \: \textrm{is on list} \big| T_{\bu\bs\bx\by}, m=1 , {\calE_c(1)}^c  \Big\}$ in \eqref{N_bar}. Therefore, the coding rate $R$ is found outside the clipping function $|\cdot|^+$ in \eqref{E2list} unlike \eqref{E_2_erasure}. This implies that
$\bar{N}_I$ might be greater than unity as expected.

\end{proof}

\section{Discussion}
\label{sec.discussion}

In this paper, we proposed universally achievable error exponents for decoding with an erasure option and a variable size list. These results were obtained by examining a universal decoder with an erasure option, inspired by Csisz\'{a}r and K\"{o}rner's \cite[p.~176]{CK81} for DMCs. By changing the decoder's parameters, one can switch from list decoding ($T \in \reals, \alpha \in (0,1)$) to decoding with an erasure option ($T \ge 0, \alpha \ge 1$).
A similar behavior was exemplified by Forney \cite{Forney68} with the optimal decoding rule for DMCs.
The proposed decoder \eqref{dec_rule1} has a similar structure to Csisz\'{a}r and K\"{o}rner's decoder \eqref{csiszar_dec}, however, it does not depend on the coding rate $R$ which make it more general.

Setting specific values to $\alpha$ and $T$ achieves
some known results. If we take $\alpha=1$ and $T=0$ in Theorem~\ref{Thorem_erasure}, we get that $E_1(R,W,T,\alpha)=E_2(R,W,T,\alpha) = E(R,W)$, where $E(R,W)$ is the exponent achieved in \cite[Th.~3.2]{MoulinWang07} for a known channel.

In \cite[eq.(11a)]{Forney68}, Forney proposed a suboptimal decoding rule with an erasure option in which the decoder declares ``$m$'' if
\be
\label{forney_second}
\Pr(\by,\bx_m) \ge e^{NT} \Pr(\by,\bx_{m_2}) \;,
\ee
otherwise, an ``erasure'' is declared, where $\Pr(\by,\bx_{m_2})$ is the probability of the second most likely code word, and $T$ is a positive parameter. Hence, the probability of the most likely code word must be at least $e^{NT}$ times higher than the probability of any other code word given $\by$. If we think of $\max_{l,\lambda} [I(\bu_{\lambda,l,m};\by) - I^*_{US}(\lambda)]$ as the normalized logarithm of the empirical generalized a-post priori probability of message $m$ given $\by$, as stated in \cite[p.~1332]{MoulinWang07}, then by setting $\alpha=1$ and $T \ge 0$ in \eqref{dec_rule1}, we obtain an empirical version of Forney's suboptimal decoding rule \eqref{forney_second}.

Unlike \cite{MoulinWang07}, we fixed the penalty function $\psi$ and the bin-depth function $\rho$ beforehand. Clearly, taking $\rho(\hP_{\bs}) = I^*_{US}(\hP_{\bs})$ is an optimal choice since it is the lowest exponential rate which ensures a vanishing encoding error probability. Higher values of $\rho$ increases the probability of decoding error (see \cite[p.~1331]{MoulinWang07}). However, it is not clear whether $\psi=\rho$ is an optimal choice in \eqref{dec_rule1} (at least not when $\alpha=1$). If, for example, we derive the exponent $\calE_1$ in Th.~1 with a general penalty function $\psi(\cdot)$ we get that
\begin{multline}
\label{E_1_erasure}
E_1(R,W,T,\alpha,\psi) = \max_{\psi}\min_{\tP_{S}} \max_{\tP_{UX|S}} \min
\Big\{ \\
\min_{\tP_{Y|XS} : I(Y;U) - \psi(\tP_{S})  \le T} \calD(\tP_{S}\tP_{UX|S}\tP_{Y|XS} \| P_S \tP_{UX|S} W) ,
\\
\min_{\tP_{Y|XS}} \Big[  \calD(\tP_{Y|XS} \| W | \tP_{S}\tP_{UX|S})\\
+ \Big|  \frac{1}{\alpha} \big(I(Y;U) - \psi(\tP_{S}) - T\big)  + \min_{\hP_{S}}[\psi(\hP_{S})- \hat{I}(U;S)] - R   \Big|^+  \Big]  \Big\} \:, \nonumber
\end{multline}
where $\hat{I}(U;S)$ is the mutual information induced by $\hP_{S}(S) \tP_{U|S}(U|S)$.
As one can see, $\psi$ is involved in many places of the above expression and therefore cannot be easily optimized. Moreover, the argument used in \cite{MoulinWang07} to prove that $\psi=\rho$ is optimal cannot be applied here. This calls for further investigation.

We note that Forney's derivation cannot be easily applied to state dependent channels where side information is present at the transmitter. The main difficulty arises from the fact that the overall channel from $\bu$ to $\by$ is not memoryless.

\section{Appendix}

\begin{proof}[Proof of Lemma~\ref{lemma_ambiguity}]
Denote $J(\bu_{\lambda,l,m} ; \by) \eqde  I(\bu_{\lambda,l,m};\by)-I^*_{US}(\lambda)$.
Suppose that the lemma is false. Therefore, there are two vectors $\bu_{\lambda,l,m}$ and $\bu_{\hat{\lambda},\hat{l},\hat{m}}$ such that:
\be
\label{eq0}
J(\bu_{\lambda,l,m} ; \by) &>& T + \alpha|J(\bu_{\lambda',l',m'} ; \by)|^+ \quad \forall m' \neq m, \forall \: \lambda',l'   \;\;, \\
J(\bu_{\hat{\lambda},\hat{l},\hat{m}};\by) &>& T + \alpha|J(\bu_{\lambda',l',m'};\by)|^+ \quad \forall m' \neq \hat{m}, \forall \: \lambda',l' \;\;.
\ee
Hence,
\be
\label{eq1}
J(\bu_{\lambda,l,m} ; \by) &>& T + \alpha|J(\bu_{\hat{\lambda},\hat{l},\hat{m}};\by)|^+ \\
\label{eq2}
J(\bu_{\hat{\lambda},\hat{l},\hat{m}};\by) &>& T + \alpha|J(\bu_{\lambda,l,m};\by)|^+ \qquad.
\ee
From \eqref{eq1}-\eqref{eq2} it is clear that
\be
T + \alpha \big[J(\bu_{\lambda,l,m}; \by) \big] <  J(\bu_{\hat{\lambda},\hat{l},\hat{m}};\by) <
\frac{1}{\alpha} \big[J(\bu_{\lambda,l,m} ; \by)- T \big] \;,
\ee
which implies that
\be
\label{eq3}
\alpha T + \alpha^2 J(\bu_{\lambda,l,m} ; \by)  &<&
J(\bu_{\lambda,l,m} ; \by)- T    \nonumber \\
0 &<& (1-\alpha^2)J(\bu_{\lambda,l,m} ; \by) - T (1+\alpha) \;.
\ee
Clearly, the right hand side of \eqref{eq3} cannot be positive since $1-\alpha^2 \le 0$, $T(1+\alpha) \ge 0$ and $J(\bu_{\lambda,l,m} ; \by)$ is positive following \eqref{eq0} . Hence, the assumption that the decoding rule is ambiguous is wrong. Note that the same proof can be used to show that for $\lambda \ge 1$, Csisz\'{a}r and K\"{o}rner's decoder \cite[Th.~5.11]{CK81} is unambiguous.
\end{proof}



\end{document}